\documentclass[11pt]{article}

\usepackage{amsthm}

\newcommand{\ignore}[1]{}

\newtheorem{theorem}{Theorem}

\newtheorem{lemma}{Lemma}

\renewcommand{\Pr}{{\bf Pr}}

\begin{document}

\title{Lower Bounds for Cover-Free Families}
\author{ Ali Z. Abdi \\ Convent of Nazareth High School \\
Grade 12, Abas 7, Haifa \and Nader H. Bshouty
\\ Dept. of Computer Science\\ Technion, Haifa, 32000}


%
\maketitle

\begin{abstract}
Let ${\cal F}$ be a set of blocks of a $t$-set $X$. $(X,{\cal F})$ is called
$(w,r)$-cover-free family ($(w,r)-$CFF) provided that, the intersection of any $w$ blocks in ${\cal F}$
is not contained in the union of any other $r$ blocks in ${\cal F}$.

We give new asymptotic lower bounds for the number
of minimum points $t$ in a $(w,r)$-CFF when $w\le r=|{\cal F}|^\epsilon$ for some constant $\epsilon\ge 1/2$.
\end{abstract}

\section{Introduction}
Let ${\cal F}$ be a set of blocks (subsets) of a $t$-set $X$. $(X,{\cal F})$ is called
$(w,r)$-cover-free family ($(w,r)-$CFF) provided that, for any $w$ blocks  $A_1,A_2,\ldots,A_w\in {\cal F}$ and any other $r$ blocks $B_1,B_2,\ldots,B_r\in {\cal F}$ we have
$$\bigcap_{i=1}^w A_i\not\subseteq \bigcup_{j=1}^r B_j.$$
Since using De Morgan a $(w,r)-$CFF can be turned into $(r,w)-$CFF, throughout the
paper we assume that $w\le r$.
Cover-free families were first introduced in 1964 by Kautz and Singleton~\cite{KS64}.

Let $N(n,(w,r))$ denote the minimum number of points $|X|$ in any $(w,r)$-CFF
having $|{\cal F}|=n$ blocks. The best known lower bound for $N(n,(1,r))$ is~\cite{DR82,F96,R94}
\begin{eqnarray}\label{bound1}
N(n,(1,r))=\Omega\left(\frac{r^2}{\log r}\log n\right)
\end{eqnarray} when $r\le \sqrt{n}$ and $\Omega(n)$ when $r>\sqrt{n}$.
The constant of the $\Omega()$ is asymptotically $1/2$, $1/4$ and $1/8$, respectively.
Stinson et. al, \cite{SWZ00}, proved that
\begin{eqnarray}\label{base1}
N(n,(w,r))\ge N(n-1,(w-1,r))+N(n-1,(w,r-1)).
\end{eqnarray} They then use it with (\ref{bound1})
to prove two bounds. The first bound
is
\begin{eqnarray}\label{fb}
N(n,(w,r))\ge \Omega\left( \frac{{w+r\choose w}(w+r)}{\log {w+r\choose w}}\log n\right)
 \end{eqnarray} when $r\le \sqrt{n}$,~\cite{SWZ00,MW04}, and
\begin{eqnarray}\label{bst}
N(n,(w,r))\ge \Omega\left( \frac{{w+r\choose w}}{\log {(w+r)}}\log n\right)
\end{eqnarray} for
any $r\le n$,~\cite{SWZ00}.
To the best of our knowledge (\ref{bst}) is the best bound known when $\sqrt{n}\le r\le n$.
D'yachkov et. al. breakthrough result,~\cite{DVPS14}, implies that for $r\le \sqrt{n}$ and $r,n\to\infty$
\begin{eqnarray}\label{fb2}
N(n,(w,r))= \Theta\left( \frac{{w+r\choose w}(w+r)}{\log {w+r\choose w}}\log n\right)
\end{eqnarray}
and for $r\ge \sqrt{n}$ and $r,n\to\infty$
\begin{eqnarray}\label{bst2}
N(n,(w,r))\le O\left( \frac{r}{w}\cdot\frac{{w+r\choose w}}{\log {(w+r)}}\log n\right).
\end{eqnarray}

In this paper we give a new lower bound for $(w,r)$-CFF when $r>\sqrt{n}$.
We combine the two techniques used in \cite{SWZ00,MW04} and \cite{AA05} to give the following asymptotic lower bound.

\begin{theorem}\label{T12}
For any $2\le k\le w<r\le n/2$ and $$(n+k-1-w)^{\frac {k-1}{k}} \leq r \leq (n+k-w)^{\frac {k}{k+1}}$$
$$N(n,(w,r)) \ge  \frac {k^kk!}{2(k+1)^{2k}} \frac {r^{w+1}}{(w+1)!\ln^k r} = \Omega \left(\frac{\sqrt{k}}{e^k}\cdot\frac {r^{w+1}}{(w+1)!\ln^{k+1} r}\log n\right)$$
and for $$r=\Omega\left(({n\log n})^{\frac{w}{w+1}}\right)$$
$$N(n,(w,r)) = \Theta\left( {n\choose w}\right).$$
\end{theorem}

Our bound is
$$\Theta\left(\frac{\sqrt{k}\cdot r}{w(e\ln r)^k}\right)$$ times greater than the previous bound
in (\ref{bst}). In particular, when $k$ is constant, our lower bound improves the bound in
(\ref{bst}) to
\begin{eqnarray}\label{bst3}
N(n,(w,r))\ge \Omega\left( \frac{r}{w\log^k r}\cdot \frac{{w+r\choose w}}{\log  {(w+r)}}\log n\right).
\end{eqnarray}

A slightly better bound can be achieved when $(n+k-w)^{\frac{k}{k+1}}\le r\le (n+k-w)^{\frac{k}{k+1}}\ln^{1/(k+1)}n$.

For example, let $w=4$. The table in Figure~\ref{exa} compares our
results with the previous results (asymptotic values)

\begin{figure}[h!]
\begin{center}
\begin{tabular}{c|l|l|l|}
& Previous Lower   &Upper & Our Lower\\
$r$ & Bounds (\ref{fb}), (\ref{bst}) & Bound~\cite{DVPS14} &Bound\\
\hline \hline
$r\le n^{1/2}$ & ${r^5}\frac{\log n}{\log r}$ & $r^5\frac{\log n}{\log r}$& -----\\
\hline
$n^{1/2}\le r\le n^{2/3}$ & ${r^4}\frac{\log n}{\log r}$& $r^5\frac{\log n}{\log r}$&${r^5}\frac{\log n}{\log^3 r}$\\
\hline
$n^{2/3}\le r\le n^{3/4}$ & ${r^4}\frac{\log n}{\log r}$& $r^5\frac{\log n}{\log r}$&${r^5}\frac{\log n}{\log^4 r}$\\
\hline
$n^{3/4}\le r\le n^{4/5}$ & ${r^4}\frac{\log n}{\log r}$& $r^5\frac{\log n}{\log r}$&${r^5}\frac{\log n}{\log^5 r}$\\
\hline
$ n> r\ge (n\log n)^{4/5}$ & $r^4$& $n^4$&$n^4$\\
\hline
\end{tabular}
\end{center}
\caption{Results for $w=4$.}
\label{exa}
\end{figure}

\section{First Lower Bound}
In this section we prove
\begin{lemma}\label{L11} Let $w\le r\le n/2$. If
$$r=\Omega\left(\left(n\log n\right)^{\frac{w}{w+1}}\right)$$
then
\begin{eqnarray}
N(n,(w,r))=\Theta\left( {n\choose w}\right).
\end{eqnarray}
Otherwise,
\begin{eqnarray}
N(n,(w,r))\ge \Omega\left(\left(\frac{r}{(w+1)\ln r}\right)^{w+1}\log n\right).
\end{eqnarray}
\end{lemma}

Lemma~\ref{L11} follows from the following
\begin{lemma}\label{T1} Let $\epsilon<1$ be any constant.
For $w\le r\le n/2$ we have
\begin{eqnarray}
N(n,(w,r))\ge \min\left((1-\epsilon)\frac{w^w}{(w+1)^{2w+1}}\cdot \frac{r^{w+1}}{\ln^wr}\ \ \  ,
\ \ \  \epsilon{n\choose w}\right)
\end{eqnarray}
\end{lemma}
\begin{proof} Let $(X,{\cal F})$ be an optimal $(w,r)$-CFF. Let ${\cal F}=\{F_1,\ldots,F_n\}$, $|X|=N=N(n,(w,r))$ and assume without loss
of generality that $X=[N]:=\{1,\ldots,N\}$.
Define $v^{(i)}\in\{0,1\}^n$, $i=1,\ldots,N$ where $v^{(i)}_j=1$ if and only if
$i\in F_j$. Let $V=\{v^{(i)}|i=1,\ldots,N\}$. Let $V_0$ be the set of $v^{(i)}$ of weight $wt(v^{(i)})$ (i.e., $\sum_jv^{(i)}_j$) equal to $w$. Let
$$m=\frac{(w+1)^2n\ln r}{wr}$$
and consider the two sets $V_1=\{v^{(i)}\ |\ w<wt(v^{(i)})< m\}$ and  $V_2=\{v^{(i)}\ |\ wt(v^{(i)})\ge m\}$. Obviously, $V=V_0\cup V_1\cup V_2$ is a partition of $V$. Suppose
$$|V_0|\le \epsilon{n\choose w}$$ and
$$\max(|V_1|,|V_2|)\le (1-\epsilon)\frac{w^w}{(w+1)^{2w+1}}\cdot \frac{r^{w+1}}{\ln^wr}.$$

Consider $W=\{(j_1,\ldots,j_w)\ |\ 1\le j_1<\cdots<j_w\le n\}$ and $W'\subset W$
the set of all $(j_1,\ldots,j_w)$ where no $v^{(i)}\in V_0$, $i=1,\ldots,N$, satisfies $v_{j_1}^{(i)}=\cdots=v^{(i)}_{j_w}=1$. Obviously, $$|W'|={n\choose w}-|V_0|\ge (1-\epsilon){n\choose w}.$$
Fix an element $v\in V_1$ and randomly and uniformly choose
$j=(j_1,\ldots,j_w)\in W'$. We have
\begin{eqnarray*}
\Pr_{j\in W'}[v_{j_1}=\cdots=v_{j_w}=1] &\le& \frac{{wt(v)\choose w}}{|W'|}
\le\frac{{m\choose w}}{(1-\epsilon){n\choose w}}.
\end{eqnarray*}
Therefore, the expectation of the number of $v\in V_1$ for which $v_{j_1}=\cdots=v_{j_w}=1$ is at most
\begin{eqnarray*}
\frac{{m\choose w}|V_1|}{(1-\epsilon){n\choose w}}&\le & \frac{1}{1-\epsilon}\left(\frac{m}{n}\right)^w|V_1|\\
&\le& \frac{1}{1-\epsilon}\frac{(w+1)^{2w}\ln^wr}{w^wr^w}\cdot (1-\epsilon)\frac{w^w}{(w+1)^{2w+1}}\cdot \frac{r^{w+1}}{\ln^wr}\\
&=& \frac{r}{w+1}.
\end{eqnarray*}
Therefore, there is $j'=(j_1',\ldots,j_w')\in W'$ such that the number of $v\in V_1$
that satisfies $v_{j_1'}=\cdots=v_{j_w'}=1$ is $r_1\le r/(w+1)$. Since the weight of every $v\in V_1$ is greater than $w$, we can choose $r_1$ new entries $j_1'',\ldots,j_{r_1}''\not\in \{j_1',\ldots,j_w'\}$ such that for every $v\in V_1$ where $v_{j_1'}=\cdots=v_{j_w'}=1$ there is $j_\ell''$ such that $v_{j_\ell''}=1$.

Now randomly and uniformly choose $$r_2:=\left\lceil \frac{wr}{w+1}\right\rceil$$
distinct $k_1,\ldots,k_{r_2}\in [n]$. Let $A$ be the event that $\{k_1,\ldots,k_{r_2}\}\cap \{j_1',\ldots,j_w'\}\not=\O$.
The probability that $A$ does not happen is
\begin{eqnarray*}
\frac{{n-w\choose r_2}}{{n\choose r_2}}\ge  \frac{{n-w \choose r_2}}{2^w {n-w\choose r_2}}=\frac{1}{2^w}
\end{eqnarray*}

Then
\begin{eqnarray*}
\Pr[A\vee(\exists v\in V_2)\ v_{k_1}=\cdots=v_{k_{r_2}}=0]&\le& 1-\frac{1}{2^w}+|V_2|\frac{{n-m\choose r_2}}{{n\choose r_2}}\\
&\le& 1-\frac{1}{2^w}+|V_2| \left(\frac{n-m}{n}\right)^{r_2}\\
&\le& 1-\frac{1}{2^w}+|V_2|e^{-\frac{mr_2}{n}}
\end{eqnarray*}
and
\begin{eqnarray*}
|V_2|e^{-\frac{mr_2}{n}}&\le&(1-\epsilon)\frac{w^w}{(w+1)^{2w+1}}\cdot \frac{r^{w+1}}{\ln^wr}\cdot e^{-\frac{(w+1)^{2}\ln r}{wr}r_2 }\\
&\le &  (1-\epsilon)\frac{w^w}{(w+1)^{2w+1}}\cdot \frac{r^{w+1}}{\ln^wr}\cdot e^{-(w+1)\ln r}\\
&= &  (1-\epsilon)\frac{w^w}{(w+1)^{2w+1}}\cdot \frac{1}{\ln^wr}\\
&<  &  \frac{1}{2^w}
\end{eqnarray*}
Therefore,
\begin{eqnarray*}
\Pr[A\vee(\exists v\in V_2)\ v_{k_1}=\cdots=v_{k_{r_2}}=0]<1.
\end{eqnarray*}

Therefore, there is $\{k_1,\ldots,k_{r_2}\}$ such that
$\{k_1,\ldots,k_{r_2}\}\cap \{j_1',\ldots,j_w'\}=\O$ and for every $v\in V_2$ there
is $k_\ell \in \{k_1,\ldots,k_{r_2}\}$ where $v_{k_\ell}=1$.

Now it is easy to see that there is no $v\in V$ where $v_{j_1'}=\cdots=v_{j_w'}=1$,
$v_{j_1''}=\cdots=v_{j_{r_1}''}=0$ and $v_{k_1}=\cdots=v_{k_{r_2}}=0$. This implies that
$$\bigcap_{i=1}^w F_{j_i'}\subseteq \bigcup_{i=1}^{r_1} F_{j_i''}\cup\bigcup_{i=1}^{r_2} F_{k_i}$$ which is a contradiction.
\end{proof}

\section{The Second Bound}\label{s2}

In this section we prove Theorem~\ref{T12}.

\begin{lemma}
For any $2\le k\le w\le r\le n/2$ and $$ 2 \leq r \leq (n+k-w)^{\frac {k}{k+1}}$$
$$N(n,(w,r)) \ge  \frac {k^kk!}{2(k+1)^{2k}} \frac {r^{w+1}}{(w+1)!\ln^k r} = \Omega \left(\frac {r^{w+1}}{(w+1)!\ln^k r}\right).$$
\end{lemma}
\begin{proof}
We prove the lemma by induction on $w$.

From Lemma~\ref{T1} the lemma holds for $w=k$. Now assume the bound holds for some $w$ and every $r$ that satisfies $r \leq (n+k-w)^\frac{k}{k+1}$. We now prove the bound for $w+1$ and $r\le (n+k-w-1)^\frac{k}{k+1}$
\begin{eqnarray}
N(n,(w+1,r)) &\ge & N(n-1,(w,r))) + N(n-1,(w+1,r-1))\label{I5} \\
&\ge & \sum\limits_{j=1}^{r} N(n-r+j-1,(w,j))\label{I6} \\
&\ge & N(n-r,(w,1)) +\nonumber \\ & &\ \ \ \ \ \ \sum\limits_{j=2}^{r} \frac {k^kk!}{2(k+1)^{2k}} \frac {j^{w+1}}{(w+1)!\ln^k j} \label{I7}\\
&\ge & \frac {k^kk!}{2(k+1)^{2k}(w+1)!\ln^k r} \sum\limits_{j=1}^{r} j^{w+1} \nonumber \\
&\ge & \frac {k^kk!}{2(k+1)^{2k}(w+1)!\ln^k r} \int_0^r \! x^{w+1} \, \mathrm{d}x \nonumber \\
&\ge&  \frac {2k^kk!}{(k+1)^{2k}} \frac {r^{w+2}}{(w+2)!\ln^k r} \nonumber
\end{eqnarray}

Here, inequality (\ref{I5}) comes from \cite{SWZ00}.
Inequality (\ref{I6}) follows from the fact that $N(n-r+1,(w+1,1))\ge N(n-r,(w,1))$.
Inequality (\ref{I7}) follows from the induction hypothesis since
\begin{eqnarray*}
j&=&r-(r-j)\\
&\le& (n+k-w-1)^\frac{k}{k+1}-(r-j)\\
&\le& (n+k-w-1-(r-j))^\frac{k}{k+1}\\
&=& ((n-r+j-1)+k-w)^\frac{k}{k+1}.
\end{eqnarray*}
\end{proof}


\begin{thebibliography}{}
\bibitem{AA05}
N. Alon, V. Asodi.
Learning a Hidden Subgraph.
{\it SIAM J. Discrete Math.} 18(4). pp. 697--712 (2005).

\bibitem{DR82}
A. G. D'yachkov and V. V. Rykov.
Bounds on the length of disjunctive codes.
{\it Problemy Peredachi Informatsii}, 18(3), pp. 7--13, (1982).

\bibitem{DVPS14}
A. G. D'yachkov, I. V. Vorob'ev, N. A. Polyansky, V. Yu. Shchukin.
Bounds on the rate of disjunctive codes.
Problems of Information Transmission.
50(1), pp. 27--56. (2014).

\bibitem{F96}
Z. F\"{u}redi.
On $r$-Cover-free Families.
{\it J. Comb. Theory, Ser. A}, 73(1). pp. 172--173. (1996).

\bibitem{KS64}
W. H. Kautz and R. C. Singleton.
Nonrandom binary superimposed codes, {\it IEEE Trans.
Inform. Theory}. 10, pp. 363--377. (1964).

\bibitem{MW04}
X. Ma and R. Wei.
On Bounds of Cover-Free Families.
Designs, Codes and Cryptography, 32, pp. 303--321, (2004).

\bibitem{R94}
M. Ruszink\'{o}.
On the Upper Bound of the Size of the $r$-Cover-Free Families.
{\it J. Comb. Theory, Ser. A}. 66(2). pp. 302--310. (1994).

\bibitem{SWZ00}
D. R. Stinson, R. Wei and L. Zhu.
Some New Bounds for Cover-Free Families.
{\it Journal of Combinatorial Theory, Series A}. 90, pp. 224--234 (2000).

\end{thebibliography}
\end{document}